\documentclass[amsmath,amssymb]{revtex4-1}

\usepackage[utf8x]{inputenc}
\usepackage{ucs}

\usepackage{amsfonts}
\usepackage{dcolumn}
\usepackage{bm}
\usepackage[usenames]{color}
\usepackage{multirow}
\usepackage{graphicx}
\usepackage{hyperref}
\usepackage{subfig}
\usepackage{booktabs}
\usepackage{xcolor}
\usepackage[normalem]{ulem}

\usepackage{ucs}
\usepackage{amssymb}
\usepackage{makeidx}
\usepackage{diagbox}
\usepackage{tabularx}

\usepackage{amsthm}

\newtheorem*{corollary}{Corollary}
\newtheorem*{lemma}{Lemma}

\usepackage{float} 
\usepackage{wrapfig} 

\usepackage{upgreek} 
\usepackage{cancel} 
\usepackage{mathdots} 
\usepackage{mathrsfs} 
\graphicspath{{figures/}} 

\begin{document}


\title{Biparametric complexities and the generalized Planck radiation law}

\author{David Puertas-Centeno$^{1}$, I. V. Toranzo$^{1}$ and
J. S. Dehesa$^{1}$}

\address{
  $^{1}$Instituto Carlos I de F\'{\i}sica Te\'orica y
Computacional and Departamento de F\'{\i}sica At\'omica Molecular
y Nuclear, Universidad de Granada, Granada 18071, Spain.
}

\begin{abstract}
Complexity theory embodies some of the hardest, most fundamental and most challenging open problems in modern science. The very term \textit{complexity} is very elusive, so that the main goal of this theory is to find meaningful quantifiers for it. In fact we need various measures to take into account the multiple facets of this term. Here some biparametric Crámer-Rao and Heisenberg-Rényi measures of complexity of continuous probability distributions are defined and discussed. Then, they are applied to the  blackbody radiation at temperature $T$ in a $d$-dimensional universe. It is found that these dimensionless quantities do not depend on $T$ nor on any physical constants. So, they have an universal character in the sense that they only depend on the spatial dimensionality. To determine these complexity quantifiers we have calculated their dispersion (typical deviations) and entropy (Rényi entropies and the generalized Fisher information) constituents. They are found to have a temperature-dependent behavior similar to the celebrated Wien's displacement law of the dominant frequency $\nu_{max}$ at which the spectrum reaches its maximum. Moreover, they allow us to gain insights into new aspects of the $d$-dimensional blackbody spectrum and about the quantification of quantum effects associated with space dimensionality.
\vskip
0.5cm


\end{abstract} 
\maketitle

\section{Introduction}

The quantum many-body systems are not merely complicated in the way that machines are complicated but they are intrinsically complex in ways that are fundamentally different from any product of design. This intrinsic complexity makes them difficult to be fully described or comprehended. Moreover, in order to substantiate our intuition that complexity lies between perfect order and perfect disorder (i.e., maximal randomness), the ultimate goal of complexity theory is to find an operationally meaningful, yet nevertheless computable, quantifier of complexity. Many efforts have been done to understand it by using concepts extracted from information theory and density functional methods (see e.g., \cite{arndt,parr89,frieden04,sen2012}). First, they used information entropies (Fisher information \cite{fisher} and Shannon, Rényi and Tsallis entropies \cite{shannon_49,renyi_70,tsallis0}) of the one-body densities which characterize the quantum states of the system. These quantities describe a single aspect of oscillatory (Fisher information) and spreading (Shannon, Rényi and Tsallis entropies) types of the quantum wavefunction. However, this is not enough to describe and quantify the multiple aspects of the complexity of natural systems from particle physics to cosmology \cite{gellmann0,gellmann,badii,sen2012,seitz}. In fact there is no general axiomatic formalization for the term \textit{complexity} (see a recent related effort \cite{rudnicki1}), but various quantifiers which take simultaneously into account two or more aspects of it. Most relevant up until now are the two-factor complexity measures  of Crámer-Rao \cite{dehesa_1,antolin_ijqc09}, Fisher-Shannon \cite{romera_1,angulo_pla08} and LMC (Lopez-ruiz-Mancini-Calvet)\cite{lopez95,catalan_pre02,anteneodo} types. They quantify the combined balance of two macroscopic aspects of the quantum probability density of the systems, and satisfy a number of interesting properties: dimensionless, bounded from below by unity \cite{dembo,guerrero}, invariant under translation and scaling transformation \cite{yamano_jmp04,yamano_pa04}), and monotone in a certain sense \cite{rudnicki1}. Later on, some generalizations of these three basic quantities have been suggested which depend on one or two parameters, such as the measures of Fisher-Rényi \cite{romera08,antolin09,romera09,antolin_ijqc09,rudnicki} and LMC-Rényi \cite{pipek,lopez,lopezr,sanchez-moreno} types. \\

This article has two goals. First, we introduce two biparametric measures of complexities for continuous probability densities, which are qualitatively different from all the previously known ones, generalizing some of them (Crámer-Rao, LMC); namely, the generalized Crámer-Rao (or Fisher-Heisenberg) and the Heisenberg-Rényi measures. Then, we discuss their main properties. Second, we apply these two complexity measures to the generalized Planck radiation law, which gives the spectral frequency density of a blackbody at temperature $T$ in a $d$-dimensional universe. This quantum object has played a fundamental role since the pionnering works of Planck at the birth of quantum mechanics up until now from both theoretical \cite{cardoso,ramos1,ramos2,alnes,garcia,lehoucq,stewart,nozari,zeng,Tor,tsallis} and experimental \cite{ade,maia,martinez,mather0,mather} standpoints. Keep in mind e.g. that the cosmic microwave background radiation which baths our universe today is known to be the most perfect blackbody radiation ever observed in nature, with a temperature of about 2.7255(6) Kelvin \cite{mather0,mather,bennett,ade1}. Beyond the temperature, we will focussed on the dependence of the complexity quantities on the space dimensionality $d$; mainly, because this variable is crucial in the analysis of the structure and dynamics of natural systems and phenomena from condensed matter to high energy physics, cosmology and quantum infomation (see e.g. \cite{dehesa2,acharyya,brandon,krenn,march,spengler,rybin,salen} and the monographs \cite{weinberg,herschbach,dong,dehesa1}). The $d$-dependence of the entropy-like and complexity-like quantities of the $d$-dimensional hydrogenic and harmonic systems has been recently reviewed \cite{dehesa2} up until 2012, and more recently the three basic complexity measures (Crámer-Rao, Fisher-Shannon and LMC) of the $d$-dimensional blackbody have been shown to have an universal character in the sense that they depend neither on temperature nor on the Planck and Boltzmann constants, but only on the space dimensionality $d$. In this work, we will prove that a similar statement can be argued for the two biparametric measures of complexity mentioned above.  \\

The structure of the article is the following. In section II some spreading quantities (typical deviations, Rényi entropy, biparametric Fisher information) of a general continuous one-dimensional probability distribution are considered, and their meanings and properties relevant to this work are briefly given and discussed. In addition, two biparametric complexity measures of Crámer-Rao and Heisenberg-Rényi character are defined in terms of the previous spreading quantities. In Section III the central moments, Rényi entropy and generalized Fisher information are studied analytically and numerically for the $d$-dimensional blackbody spectrum in terms of its temperature and the space dimensionality. This research allows to conclude that these measures could be used as quantifiers of the spatial anisotropy whose details are being investigated at present in a more precise way with the most modern astronomical tools. In particular, the generalized Fisher information (due to its strong sensitivity to the spectrum fluctuations) could contribute to the elucidation of the origin of the cosmic microwave background anisotropies.

 Then, in section IV the generalized measures of complexity of the blackbody spectrum are investigated, finding that the biparametric complexities (Crámer-Rao and Heisenberg-Rényi) of the $d$-dimensional blackbody are dimensionless and, moreover, they do not depend on the temperature $T$ of the system nor on any physical constant (e.g., Planck's constant, speed of light, Boltzmann's constant). Thus, they are universal quantities since they only depend on the spatial dimensionality.

 Finally, in section V some concluding remarks are given, and various open problems are pointed out relative to the new complexity measures as well as the frequency distribution of a tri- and $d$-dimensional blackbody in order to shed some more light on the knowledge of the radiation that baths our universe.

\section{Basic and extended measures of complexity}

In this Section first we briefly give the three basic complexity measures of a probability distribution; namely, the Crámer-Rao, Fisher-Shannon and LMC complexities. Then, we define two novel families of complexity measures (the biparametric Crámer-Rao and Heisenberg-Rényi complexities) which generalize the previous ones.

\subsection{Basic complexities}

Let us consider a general one-dimensional random variable $X$ characterized by the continuous probability distribution $\rho(x)$,
$x \in \Lambda \subseteq \mathbb{R}$. Obviously it is asumed that the density is normalized to unity, so that $\int_{\Lambda} \rho(x) dx  =  1$. The basic measures of complexity of Crámer-Rao, Fisher-Shannon and LMC types are defined by means of the expressions
  \begin{eqnarray}
     \label{Cramer_rao}
     C_{CR}\left[\rho\right]&=& F\left[\rho\right] \, V\left[\rho\right],\\
     \label{Complejidad_de_Fisher-Shannon}
     C_{FS}\left[\rho\right]&=&\frac{1}{2 \pi e} F\left[\rho\right] \, \exp \left(2 S\left[\rho\right]\right),\\
      \label{Complejidad_LMC}
      C_{LMC}\left[\rho\right]&=& D\left[\rho\right] \, \exp \left(S\left[\rho\right]\right),
   \end{eqnarray}
  respectively. The symbols $F[\rho]$, $V[\rho]$, $S[\rho]$, and $D[\rho]$ denote the standard Fisher information \cite{fisher,frieden04}
\begin{equation}
F[\rho]=\int_{\Delta}\frac{|\rho'(x)|^{2}}{\rho(x)}dx,
\label{eq:inffish}
\end{equation}
the variance (see e.g.\cite{hall_pra99})
\begin{equation}
     \label{varianza}
     V\left[\rho\right]=\langle x^2 \rangle-\langle x \rangle^2; \quad \langle f\left(x\right) \rangle=\int_{\Delta} f\left(x\right) \rho(x) \,dx,
  \end{equation}
the Shannon entropy \cite{shannon_49}
\begin{equation}
S[\rho]=-\int_{\Delta}\rho(x)\ln[\rho(x)]dx,
\label{eq:entropshan}
\end{equation}
and the disequilibrium \cite{onicescu}
\begin{equation}
D[\rho]= \int_{\Delta}[\rho(x)]^{2}dx, 
\label{eq:disequilibrium}
\end{equation}
of the probability density $\rho(x)$, respectively. The Fisher information quantifies the gradient content or pointwise concentration of the probability over its support interval $\Lambda$. The variance, the Shannon entropy and the disequilibrium measure the following spreading properties of $\rho(x)$:  the concentration of the density around the centroid $\langle x \rangle$, the total extent to which the density is in fact concentrated, and the separation of the density with respect to equiprobability, respectively. Note that the Fisher information has a property of locality because it  is very sensitive to the fluctuations of the density, contrary to the three spreading quantities which have a global character because they are power functionals of the density. The property of locality is very important in the quantum-mechanical description of physical systems, because their associated wavefunctions are inherently oscillatory for all quantum states except at the ground case. \\  
  
Therefore, the Crámer-Rao, Fisher-Shannon and LMC complexities of $\rho(x)$ are statistical measures of complexity which quantify the combined balance of two aspects of the density described by their two associated spreading components of dispersion and entropic character. Both the Crámer-Rao and Fisher-Shannon complexities have a local-global character but in a different sense: The Crámer-Rao complexity $C_{CR}\left[\rho\right]$ quantifies the gradient content of $\rho(x)$ and the probability concentration around its centroid, and the Fisher-Shannon complexity $C_{FS}\left[\rho\right]$ measures the gradient density jointly with the total extent of the density in the support interval as given by the squared Shannon entropy power. The LMC complexity $C_{LMC}\left[\rho\right]$ has a global-global character because it measures simultaneously two global spreading aspects of $\rho(x)$: the disequilibrium and the total extent of the density as given by the Shannon entropy power. These three dimensionless complexity measures are known to be bounded from below by unity \cite{dembo,guerrero}, and invariant under translation and scaling transformation \cite{yamano_jmp04,yamano_pa04}. The question whether these quantities are minimum for the two extreme (or \textit{least complex}) distributions corresponding to perfect order and maximum disorder (associated to an extremely localized Dirac delta distribution and a highly flat distribution in the one dimensional case, respectively) is a long standing and controverted issue \cite{shiner,sanchez-moreno} which has been partially solved. Indeed, these three statistical measures have been recently shown to be monotone in a well-defined sense \cite{rudnicki1}.

\subsection{Extended complexities}
  
Now, inspired by Lutwak et al' efforts \cite{lutwak}, we introduce two generalized statistical measures of complexity of local-global character (the biparametric Crámer-Rao or Fisher-Heisenberg and Heisenberg-Rényi complexities) which extend the basic complexity measures mentioned above. For this purpose we take into account the p\textit{th}-typical deviation  (or p\textit{th} absolute deviation with respect to the middle value) $\sigma_p[\rho]$ of the probability density $\rho(x)$ defined as
\begin{equation}
\label{centmom}
 \sigma_p[\rho] =\left\{\begin{array}{c} e^{\int_{\Delta} \rho(x) \ln |x-\langle x\rangle | dx},\quad \text{if}\,\, p=0
 \\ 
 \left(\int_{\Delta}\left|x-\langle x\rangle\right|^p \rho(x)\,dx\right)^\frac1p,\quad \text{if}\,\, 0<p<\infty \\ 
  ess\sup\{|x-\langle x\rangle|:\rho(x)>0\},\quad \text{if}\,\, p=\infty\\
  \end{array}\right.
 \end{equation}
and the Rényi entropic power defined as 
\begin{equation}
\label{eq:renyilength}
N_\lambda[\rho] = e^{R_{\lambda}[\rho]},
\end{equation}
where $R_{\lambda}[\rho]$ denotes the standard or monoparametric Rényi entropy  of order $\lambda$ \cite{renyi_70} given by  \begin{equation}
R_{\lambda}[\rho]=\frac{1}{1-\lambda}\ln\left(\int_{\Delta}[\rho(x)]^{\lambda}dx\right);\quad \lambda>0,\, \lambda\neq1.
\label{eq:entropRen}
\end{equation}
Note that the the $p$-typical deviations quantify different facets (governed by the parameter $p$) of the concentration of the probability density around the centroid, and the $\lambda$-Rényi entropic powers measure various aspects (governed by $\lambda$) of the global spreading of the probability density along its support interval. In particular we have that
\begin{equation*}
 N_\lambda[\rho] =\left\{\begin{array}{c} \text{Length of the support},\quad \text{if}\,\, \lambda =0
 \\ 
 e^{-\langle\ln\rho\rangle},\quad \text{if}\,\, \lambda=1 \\ 
  \langle\rho\rangle^{-1}\quad \text{if}\,\, \lambda=2
  \\
\rho_{max}^{-1},\quad \text{if}\,\, \lambda \to \infty.\\
  \end{array}\right.
 \end{equation*}
It is also worth to realize the well-known fact that, when $\lambda$ tends to unity, the Rényi entropy $R_{\lambda}[\rho]$ tends to the Shannon entropy $S[\rho]$.\\

Besides, to define the novel complexity quantifiers we need to consider the (scarcely known) biparametric ($p,\lambda$)-Fisher information \cite{lutwak} defined as
\begin{equation}
  \label{eq:genfish}
\phi_{p,\lambda}[\rho]=\left\{\begin{array}{c}ess\sup\{|\rho(x)^{\lambda-2}\rho'(x)|^\frac1\lambda:x\in\Delta\},\quad \text{if}\,\, p=1
 \\ 
 \left(\int_{\Delta} \left|[\rho(x)]^{\lambda-2}\rho'(x)\right|^{q}\rho(x)\,dx\right)^{\frac{1}{q\lambda}}
 , 1<p<\infty, \, =1\\ 
 \left(\text{Total variation of}\, \frac{\rho(x)^\lambda}{\lambda}\right)^{\frac1\lambda},\,\, p \to\infty\\
  \end{array}\right.
 \end{equation}
with $\frac{1}{p}+\frac{1}{q}=1$, $p\in(1,\infty)$,  and $\lambda\in\mathbb{R}$. Note that for the particular values $(p,\lambda)=(2,1)$, this generalized measure reduces to the standard Fisher information $F[\rho]$ in the sense that $\phi_{2,1}[\rho]^2 = F[\rho]$. It is then clear that the $(p,\lambda)$-Fisher informations quantify various fluctuation-like facets (governed by the parameters $p$ and $\lambda$) of the probability density $\rho(x)$, including the gradient content (when $p=2$ and $\lambda=1$) . \\
  
The biparametric ($p,\lambda$)-Crámer-Rao (also called by biparametric Fisher-Heisenberg) complexity is defined as
 \begin{equation}
 \label{bcr}
 C_{CR}^{(p,\lambda)}[\rho]=\mathcal K_{CR}(p,\lambda)\,\,\phi_{p,\lambda}[\rho]\,\,\sigma_p[\rho],
  \end{equation}
 where $1\le p\le \infty$ and $\lambda>\frac 1{1+p}$, and the symbols $\sigma_p[\rho]$ and $\phi_{p,\lambda}[\rho]$ denote the typical deviation of order $p$ and the Fisher information of order ($p,\lambda$), respectively, previously defined. Moreover, the constant $\mathcal K_{CR}(p,\lambda)$ is given by
  \begin{equation}
  \label{kcr0}
 \mathcal K_{CR}(p,\lambda)=
 \frac{1}{\phi_{p,\lambda}[G]\, \,\sigma_{p}[G] }
 \end{equation}
 where the $\phi_{p,\lambda}[G]$ and $\sigma_{p}[G]$ denote the values of the ($p,\lambda$)\textit{th}-Fisher information and the $p$\textit{th}-order typical deviation of the generalized Gaussian density $G(x)\equiv G_{p,\lambda}(x)$ defined as \cite{lutwak} 
 \begin{equation}\label{eq:genGauss}
  G(x)=a_{p,\lambda}\, e_\lambda(|x|^p)^{-1}
 \end{equation}
for $p \in [0,\infty]$ and $\lambda >1-p$. The symbol $e_\lambda(x)$ denotes the modified $\lambda$-exponential function:
\begin{equation}
\label{explambda}
e_\lambda(x)=(1+(1-\lambda)x)_ +^{\frac1{1-\lambda}},
\end{equation}
where the notation $t_{+} = \max\{t,0\}$ for any real $t$ has been used. Note that for $\lambda \to 1$ it reduces to the standard exponential one, $e_{1}(x)\equiv e^{x}$. Moreover, the normalization constant $a_{p,\lambda}$ has the value 
 \[
a_{p,\lambda} = \left\{\begin{array}{ll}
   \frac{p(1-\lambda)^{1/p}}{2B\left(\frac{1}{p},\frac{1}{1-\lambda}-\frac{1}{p}\right)} & \text{if}\,\,\lambda < 1,\\
   \frac{p}{2\Gamma(1/p)} & \text{if} \,\, \lambda =1,\\
     \frac{p(\lambda-1)^{1/p}}{2B\left(\frac{1}{p},\frac{\lambda}{\lambda-1}\right)} & \text{if} \,\, \lambda >1,
   	\end{array}\right.
 \]
where the symbol $B(a,b)=\frac{\Gamma(a)\Gamma(b)}{\Gamma(a+b)}$ denotes the known Beta function \cite{olver} and an errata has been corrected for the $(\lambda >1)$-case: it is not $B\left(\frac{1}{p},\frac{1}{1-\lambda}\right)$ as in \cite{lutwak}, but $B\left(\frac{1}{p},\frac{\lambda}{\lambda-1}\right)$ . Note that the properties of the generalized Gaussian density are carefully detailed in Sect. II-E of \cite{lutwak}; other, more recent, expressions of this generalized density function and their corresponding properties have been shown (see e.g., \cite{oikonomou, baris}). \\

On the other hand, the constant values $\phi_{p,\lambda}[G]$ and $\sigma_{p}[G]$ are given by
  \begin{equation}
  \label{eq:fishG}
  \phi_{p,\lambda}[G]=\left\{\begin{array}{cc}
  p^\frac 1\lambda a_{p,\lambda}^{\frac{\lambda-1}{\lambda}} (p \lambda +\lambda -1)^{-\frac{(1-\frac 1p)}{\lambda}},& p<\infty \\  2^{(1-\lambda)/\lambda}\lambda^{\frac{-1}\lambda},& p \to \infty
  \end{array}\right.
  \end{equation}
 and
  \begin{equation}
  \label{eq:muG}
  \sigma_{p}[G]=\left\{\begin{array}{cc}
  (p \lambda+\lambda-1)^{-1/p},&  p\in (0,\infty),\, \lambda>\frac1{1+p} \\
  e^{\frac{\lambda}{1-\lambda}}, & p=0, \, \lambda>1\\
  1 & p \to \infty\, ,
  \end{array} \right.
  \end{equation} 
  respectively. Note that the case ($p=2, \lambda=1$) corresponds to the basic Crámer-Rao measure $C_{CR}[\rho]$ given by (\ref{Cramer_rao}). From its definition (\ref{bcr}), we observe that the biparametric Crámer-Rao or Fisher-Heisenberg complexity quantifies the combined balance of a fluctuation aspect of the density (as given by the generalized Fisher information which depends on the parameters $p$ and $\lambda$; this aspect is the gradient content in the particular case $p = 2, \lambda = 1$) and a dispersion facet of the probability concentration with respect to the centroid (as given by the central moment of order $p$; this aspect is the variance of the density in the particular case $p=2$).
\\
  
The biparametric ($p,\lambda$)-Heisenberg-Rényi complexity is defined as
\begin{equation}
\label{bhr}
 C_{HR}^{(p,\lambda)}[\rho]=\mathcal{K}_{HR}(p,\lambda)\,\,\frac{\sigma_p[\rho]}{N_\lambda[\rho]}
\end{equation}
 where $ \lambda\not=1 $, $0\le p\le\infty$ and $\lambda>\frac 1{1+p}$, and the symbols $\sigma_p[\rho]$ and $N_\lambda[\rho]$ denote the $pth$-typical deviation (\ref{centmom}) and the Rényi entropic power, respectively, previously defined. Moreover, the constant $\mathcal{K}_{HR}(p,\lambda)$ has the value
\begin{equation}
\label{khr}
\mathcal{K}_{HR}(p,\lambda)=\frac{N_\lambda[G]}{\sigma_p[G]},
\end{equation}
where the symbol  $N_\lambda[G]$ denotes the Rényi entropic power of the generalized Gaussian density \cite{lutwak} is given by
\begin{equation}
  N_{\lambda}[G]= \Big(a_{p,\lambda}\,e_\lambda\left(\frac{-1}{p\lambda}\right)\Big)^{-1} 
\end{equation}
and the symbols $a_{p,\lambda}$ and $\sigma_p[G]$ have been previously given.\\


We realize from (\ref{bhr}) that the biparametric $(p,\lambda)$th-Heisenberg-Rényi complexity quantifies the combined balance of a dispersion aspect of the probability concentration with respect to the centroid (as given by $p$th-typical deviation $\sigma_p[G]$, which is the standard deviation of the density in the particular case $p=2$) and the global spreading of the density (as given by the Rényi entropic power of order $\lambda$, which boils down to the Shannon entropic power in the particular case $\lambda \to 1$).\\

These two biparametric statistical complexities turns out to be invariant under scaling and translation transformations and lower-bounded by unity, as implicitly shown in \cite{lutwak}; moreover, the equality to unity occurs at the generalized Gaussian densities given by (\ref{eq:genGauss}).\\

To get a further insight into the type of densities $G_{p,\lambda}(x)$ which minimize the two previous families of extended complexities, we have indicated in Fig. \ref{fig:Gaussian_plane} the kind of relevant distributions which correspond to a large set of values for the parameters $(p, \lambda)$. Let us only mention the standard Gaussian distribution, the exponential, the $q$-exponential, the linear, the Cauchy, the logarithmic and the ladder distributions which are particular cases of generalized Gaussian distributions with $(p, \lambda) = (2,1), (1,1), (1,q), (1,2)$, $(0,2), (2,0)$ and $(\infty,\lambda)$, respectively. More important is to note that the behavior of the tail of the distribution is closely related to $\lambda$, so that the minimizer distribution of the complexities correspond to a compact-support distribution, a light-tailed distribution (i.e., one with infinite support and all its moments finite) and a heavy-tailed distribution for the cases $\lambda>1$, $\lambda=1$ and $\lambda<1$, respectively. Thus, since Gaussianity occurs for minimal complexities, the two novel measures of complexity provide a relevant information about the relative behavior of different regions of the distribution. This is illustrated elsewhere for some specific quantum systems of Coulombic and harmonic character. Here we show in the next section the usefulness of these measures of complexity by evaluating them for the generalized Planck distribution which governs the distribution of radiation frequencies of a blackbody at temperature $T$ in a universe of arbitrary dimensionality.

\begin{figure}[H]
  \centering
 \includegraphics[scale=.33]{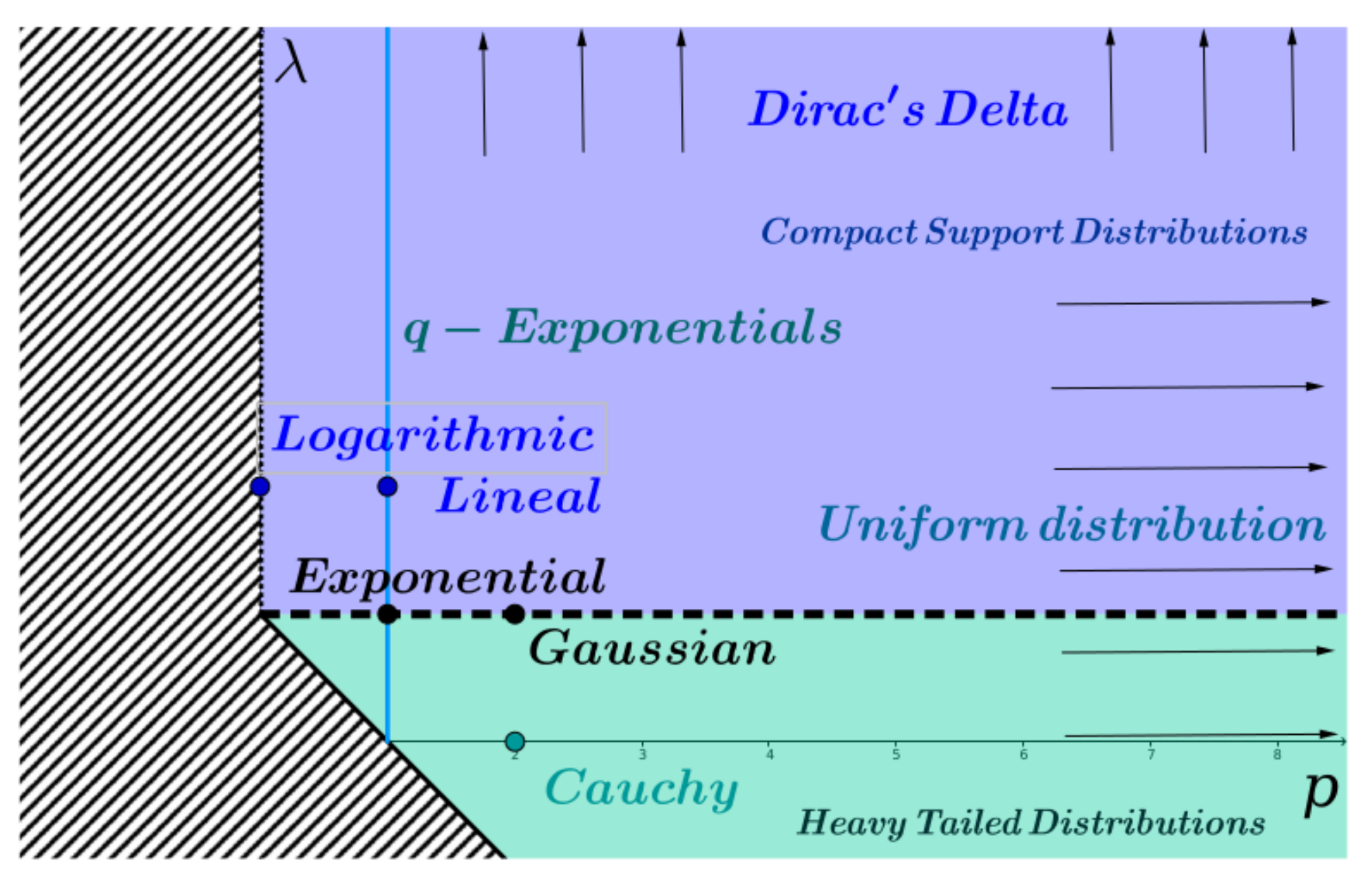}
\caption{The Gaussian $(p,\lambda)$ plane. }
\label{fig:Gaussian_plane}
\end{figure}

\section{Application to the generalized Planck radiation law}

In this section we extend the information-theoretic study of a $d$-dimensional ($d > 1$) blackbody at temperature $T$, initiated last year \cite{Tor}, by calculating the measures of dispersion (typical deviations of order $p$) beyond the standard one (i.e., that with $p=2$), the spreading quantities given by the Rényi entropic powers (which include the Shannon entropic power as a particular case), the generalized Fisher informations (which includes the standard Fisher information as a particular case) and the two biparametric complexity measures introduced in the previous section (which generalize the three basic measures of complexity mentioned above) of its 
spectral energy density  $\rho_{T}^{(d)}(\nu)$ (i.e., the energy per frequency and volume units contained in the frequency interval ($\nu$, $\nu + d\nu$) inside a $d$-dimensional enclosure maintained at temperature $T$), which is given  by the (normalized-to-unity) generalized Planck radiation law \cite{cardoso,ramos1} (see also \cite{Tor})
\begin{equation}
 \label{densidaduno}
\rho_{T}^{(d)}(\nu) = \frac{1}{\Gamma(d+1)\zeta(d+1) }\left(\frac{h}{k_{B}T}\right)^{d+1}\frac{\nu^{d}}{e^{\frac{h \nu }{k_{B} T}}-1}, 
\end{equation}
where $h$ and $k_B$ are the Planck and Boltzmann constants, respectively, and $\Gamma(x)$ and $\zeta(x)$ denote the Euler's gamma function and the Riemann's zeta function\cite{olver}, respectively.

\subsection{Typical deviations} 
Let us first determine the typical deviations $\sigma_p[\rho_{T}^{(d)}]$ of the $d$-dimensional blackbody density $\rho_{T}^{(d)}(\nu)$ defined as
\begin{equation}
 \label{eq:moments0}
 \sigma_p[\rho_{T}^{(d)}]^p = A\int_0^\infty \left|\nu-\langle\nu\rangle\right|^p\frac{\nu^d}{e^{a\nu}-1}\,d\nu,
 \end{equation}
with the notation
\[
A=\frac{1}{\Gamma(d+1)\,\zeta(d+1)}a^{d+1}, \quad \quad a=\frac h {k_BT}.
\]
Since the centroid of the density has the value
\[
\langle \nu\rangle=(d+1)\frac{\zeta(d+2)}{\zeta(d+1)}\frac 1 a \equiv \frac ba,
\]
we obtain that the typical deviation of even order $p$ of the blackbody depends on temperature $T$ as
\begin{equation}
 \label{eq:moments}
 \sigma_p[\rho_{T}^{(d)}] = (A_H(p,d))^\frac1p\, \frac{k_BT}{h},
 \end{equation}
where the proportionality constant is given by
 \begin{eqnarray}
 \label{eq:ctemoments}
 A_H(p,d) &=& \sum_{n=0}^p (-1)^{p-n}{ p \choose n}\left((d+1)\frac{\zeta(d+2)}{\zeta(d+1)}\right)^{p-n} 
          \frac{\Gamma(d+n+1)}{\Gamma(d+1)}\, \frac{\zeta(d+n+1)}{\zeta(d+1)}\nonumber\\
          & & \equiv  \sum_{n=0}^p \gamma_n(p,d)\, \zeta(d+n+1),  
\end{eqnarray}
which only depends on the space dimensionality $d$.  We observe that all $p$-typical deviations follow a Wien-like law, in the sense that they are directly proportional to the temperature of the system. In Fig. \ref{fig:3} we plot the $p$-dependence of $\sigma_p[\rho_{T}^{(d)}] \frac{h}{k_BT}$ for various dimensionalities of the universe, finding a linearly increasing behavior when $p$ is augmenting. Moreover, we note that the increasing of the space dimensionality provokes a larger dispersion of the radiation frequencies with respect to the middle value. \\

\begin{figure}[H]
  \centering
 \includegraphics[scale=.7]{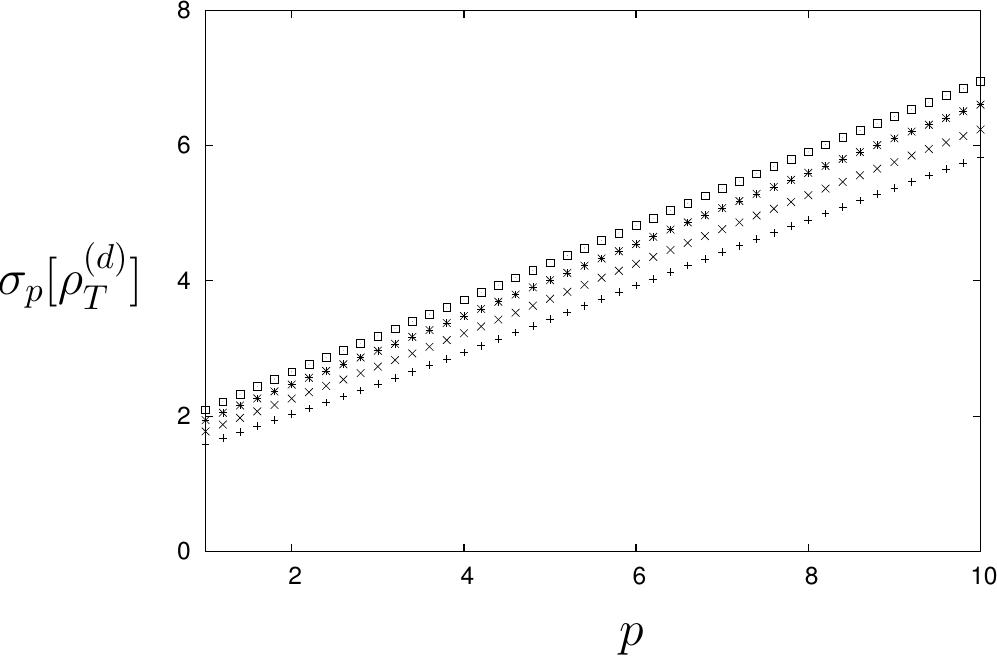}
\caption{Dependence of the p\textit{th}-typical deviation, $\sigma_{p}[\rho_T^{(d)}]$ in $\frac{h}{k_BT}$-units, on the parameter $p$  for the universe dimensionalities $ d=3(+),4(\times),5(*),6(\square)$.}
  \label{fig:3}
\end{figure}

\subsection{Rényi entropies}
 Let us now calculate the Rényi entropic power $N_\lambda[\rho_{T}^{(d)})] = e^{R_{\lambda}[\rho_{T}^{(d)}]}$, given by (\ref{eq:renyilength}), of the multidimensional blackbody density (\ref{densidaduno}) at temperature $T$, where the $\lambda$-Rényi entropy is given by
\begin{equation}
R_{\lambda}[\rho_{T}^{(d)}]=\frac{1}{1-\lambda}\ln\left(\int_{\Delta}[\rho_{T}^{(d)}]^{\lambda}d\nu\right);\quad \lambda>0,\, \lambda\neq1.
\label{eq:entropRen2}
\end{equation}
Taking into account Eqs. (\ref{eq:renyilength}) and (\ref{eq:entropRen}) and the corollary of the Lemma proved in Appendix \ref{Lemma:app}, we obtain that

\begin{equation}
\label{eq:power_renyi} 
N_{\lambda}[\rho_{T}^{(d)}]=A_R(\lambda,d)^\frac1{1-\lambda}\,\frac{k _BT}{h}, 
\end{equation}
with $ \lambda >0$, $\lambda \not=1$, and the proportionality constant 
\begin{equation}
\label{eq:AR}
A_R(\lambda,d)=\frac{\Gamma (\lambda d+1)\,\zeta_\lambda(\lambda d+1,\lambda)}{\left[\Gamma(d+1)\, \zeta(d+1)\right]^\lambda },
\end{equation}
where the symbol $\zeta_{n}(s,a) \equiv \zeta_n(s,a|1,...,1)$ denotes the modified Riemann zeta function or Barnes zeta function \cite{barnes,choi}, defined for $n \in \mathbb{N}$, which for $a \not= 0,-1,-2,...$ is known to have the integral representation, when $\text{Re}(s)>n$: 
\begin{eqnarray*}
\zeta_n(s,a)&=&\frac 1{\Gamma(s)}\int_0^\infty \frac{x^{s-1}e^{(n-a)x}}{(e^x-1)^n}\,dx\\
&=&\sum_{j=0}^{n-1}q_{n,j}(a)\,\zeta(s-j,a),	
\end{eqnarray*}

with the coefficients \cite{choi}
\begin{equation}
\label{choicoeff}  
q_{n,j}(a)=\frac 1{(n-1)!} \sum_{l=j}^ {n-1}(-1)^ {n+l-1}{l\choose j}S_{n-1}^{(l)}(1-a)^{l-j},
\end{equation}
where $S_n^{(l)}$ are the well-known Stirling's numbers of the first kind. The symbol $\zeta(s,a)$ denotes the known Hurtwitz's zeta function \cite{olver} so that for $a=1$ it boils down to the standard Riemann's zeta function $\zeta(s)$. Furthermore, it is shown in Appendix \ref{Lemma:app} that  the Barnes' zeta function can be expressed as
\begin{eqnarray*}
\zeta_n(s,a)&=&\sum_{j=0}^{n-1}q_{n,j}(a)\,\zeta(s-j,a-1)=\sum_{j=0}^{n-1}q_{n,j}(a)\,\zeta(s-j),	
\end{eqnarray*}
as far as $a\in\mathbb{N}$. In general the Barnes function is also known as the multiple (or $n$th-order) Hurwitz zeta function given by
\[
 \zeta_n(s,a|\omega_1,...,\omega_n)=\sum_{k_1,...,k_n=0}^\infty \frac1{(\Omega + a)^s},
\]
with $\text{Re}(s)>n;\,\, n \in \mathbb{N}$ and where  $\Omega=k_1\omega_1+...+k_n\omega_n$. This function was first introduced by Barnes in 1899 \cite{barnes} (who also gives the general conditions to be fulfilled by the paramaters $a$ and $\omega_i, i \in \mathbb{N}$; see also \cite{choi}) in his study of the multiple (or $n$th-order) gamma functions, whose physico-mathematical relevance was discovered in 1980 on the study about the determinants of the Laplacians on the $n$-dimensional unit sphere. \\


Note from (\ref{eq:power_renyi}) that the $\lambda$th-Rényi entropic power, which has units of frequency, follows a Wien's like displacement law in the sense that it linearly depends on the blackbody temperature. In Fig. \ref{fig:renyi} we study the behavior of the $\lambda$th-Rényi entropic power, $N_\lambda[\rho_{T}^{(d)}] \frac{h}{k_BT}$, as a function of the parameter $\lambda$  for various universe dimensionalities $d=3-6$. Briefly, we observe that (i) it monotonically decreases when $\lambda$ is increasing for all dimensionalities, and (ii) it increases when $d$ is increasing for all values of $\lambda$.
\begin{figure}[H]
  \centering
 \includegraphics[scale=.8]{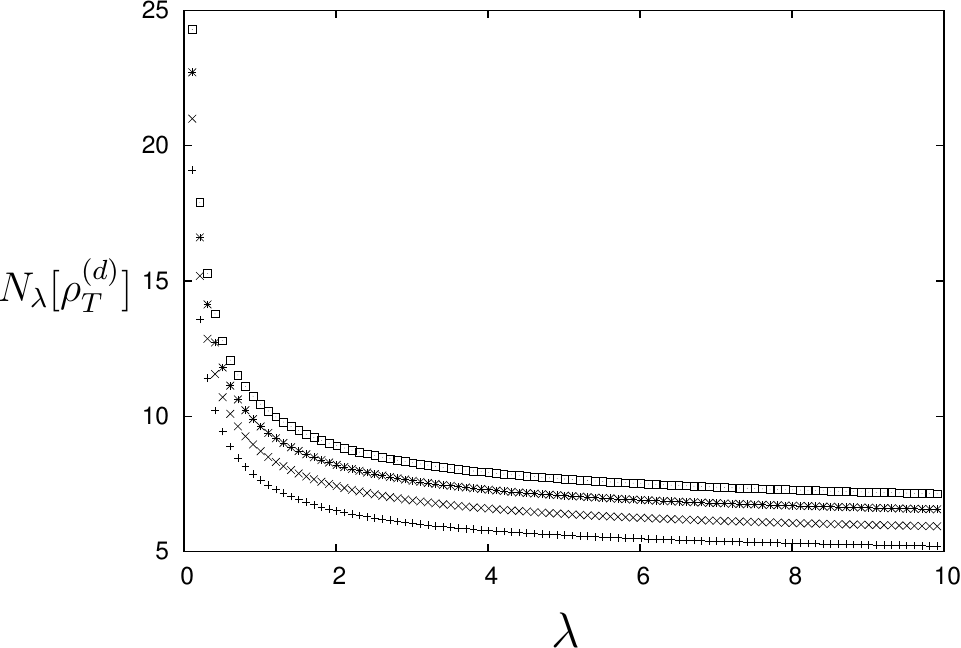}
\caption{Dependence of $\lambda$th-Rényi entropic power, $\sigma_{p}[\rho_T^{(d)}]$ in $\frac{h}{k_BT}$-units, on the parameter $\lambda$ for the universe dimensionalities $ d=3(+),4(\times),5(*),6(\square)$.}
  \label{fig:renyi}
\end{figure}

\subsection{Biparametric Fisher information}
 Now we calculate the $(p,\lambda)$th-order Fisher information of the blackbody density $\rho_{T}^{(d)}$ given by
\begin{equation}
\phi_{p,\lambda}[\rho_{T}^{(d)}] = \left(\int_{\Lambda} \left|[\rho_{T}^{(d)}]^{\lambda-2}\left(\rho_{T}^{(d)}\right)'\right|^{q}\rho_{T}^{(d)}\, d\nu\right)^{\frac{1}{q\lambda}},	
\end{equation}
where $p\in[1,\infty)$, $\frac{1}{p}+\frac{1}{q} =1$ and $\lambda > 0$. Operating similarly as before, we obtain that the biparametric Fisher information $\phi_{p,\lambda}[\rho_{T}^{(d)}]$  of the $d$-dimensional blackbody density (\ref{densidaduno}) at temperature $T$ can be expressed as
\begin{equation}
\label{eq:genfish0}
\phi_{p,\lambda}[\rho_{T}^{(d)}] = \left[A_F(p,\lambda,d)\right]^{\frac{1}{q\lambda}}\frac{h}{k_BT}, \quad \forall q \in (1,\infty), \,\,\forall \lambda >0 	
\end{equation}
with the proportionality constant 
\begin{equation}
\label{eq:cteAF}
A_F(p,\lambda,d)=\frac{I(d,q,\lambda)}{(\Gamma(d+1)\zeta(d+1))^{q\lambda -q+1}},
\end{equation}
where the symbol $I(p,\lambda,d)$ denotes the integral
\begin{equation}
\label{eq:integralI}
I(p,\lambda,d)=\int_{\mathbb{R^+}}\frac{x^{q(d\lambda -d-1)+d}}{(e^x-1)^{q\lambda +1}}\left|d(e^x-1)-xe^x\right|^q\, dx,
\end{equation}
so that for even $q$ and  $q\lambda \in \mathbb{N}$, (\ref{eq:genfish0}) one has the value
\begin{align*}
I(p,\lambda,d)&=\sum_{i=0}^q(-1)^{q+i}{q\choose i}d ^i\int_0^\infty e^{(q-i)x}\frac{x^{\alpha d-i}}{(e^x-1)^{1+q\lambda -i}}dx\nonumber \\
& = \sum_{i=0}^q(-1)^{q-i}{q\choose i}d ^i(\alpha d -i)!\zeta_{\alpha+q-i}(1+\alpha d -i ,\alpha) 
\end{align*}
with  $\alpha \equiv q\lambda -q+1$. Summarizing, we have obtained that the biparametric Fisher information, $\phi_{p,\lambda}[\rho_{T}^{(d)}]$, follows the law  (\ref{eq:genfish0}) with the proportionality constant 
\begin{equation}
\label{eq:ctefish}
A_F(p,\lambda,d) = \frac{\sum_{i=0}^q(-1)^{q-i}{q\choose i}d ^i(\alpha d -i)!\,\zeta_{\alpha+q-i}(1+\alpha d -i ,\alpha)}{[\Gamma(d+1)\,\zeta(d+1)]^{\alpha}}
\end{equation}
for even $q$ and $q\lambda \in \mathbb{N}$. Note that in the particular, standard case $\lambda=1, q=2$, one has that
\begin{equation}
A_F(2,1,d)=\frac{1}{2\zeta(d+1)}\left(\zeta(d)-\frac{d-3}{d-1}\zeta(d-1) \right),
\end{equation}
for $d>2$. Moreover, a convergence analysis of the definition (\ref{eq:genfish}) allows one to show that the generalized Fisher information $\phi_{p,\lambda}[\rho_{T}^{(d)}]$ given by (\ref{eq:genfish0}) is well-defined if and only if $\lambda p>d^*=\frac d{d-1}$ (which includes the condition $\lambda>\frac{1}{1+p}$, necessary to have finite typical deviations).\\

In Fig. \ref{fig:fisher} we plot a colour tridimensional map of biparametric Fisher information against its parameters $(q,\lambda)$, and the \textit{conjugated} representation with respect to the parameters $(p,\lambda)$ when $d=6$. Therein we observe that biparametric Fisher information has a non-trivial behaviour with an absolute minimum valley. Similar maps can be obtained for other dimensionalities. To gain more insight into it, we make two cuts in the left colour map at $p=2$ and $\lambda=2$ obtaining the two graphs (b) at the below of the figure which show a different behavior for the corresponding generalized Fisher information. In the left graph with $\lambda>1/3$ a minimum shows up at $\lambda_{min}$ for all dimensionalities $d=3-6$. In the right graph with $p>1$ we observe a monotonically decreasing behavior with respect to the parameter $p$ for all dimensionalities $d=3-6$.
\begin{figure}[H]
  \subfloat[]{%
  \begin{minipage}{\linewidth}
 \includegraphics[width=.5\linewidth]{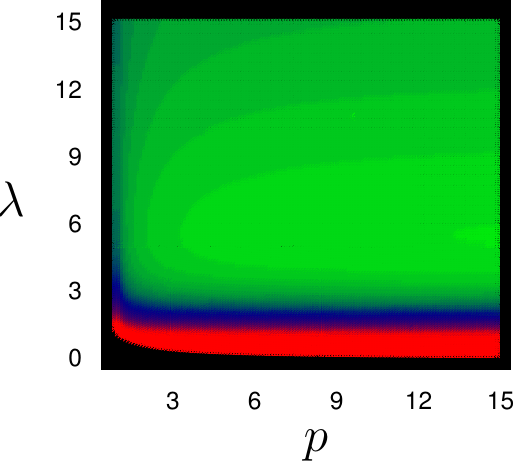}\hfill
   \includegraphics[width=.48\linewidth]{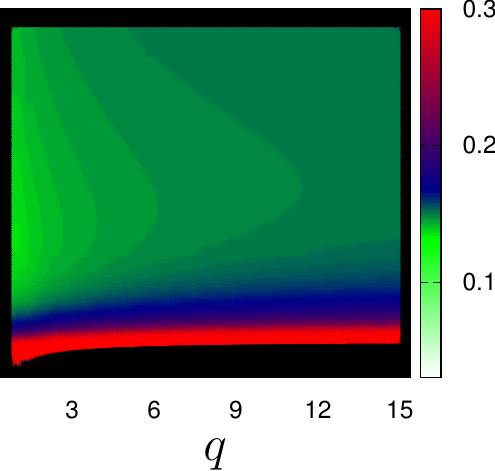}
  \end{minipage}%
  }\par
     \subfloat[]{%
      \begin{minipage}{\linewidth}
     \includegraphics[width=.5\linewidth]{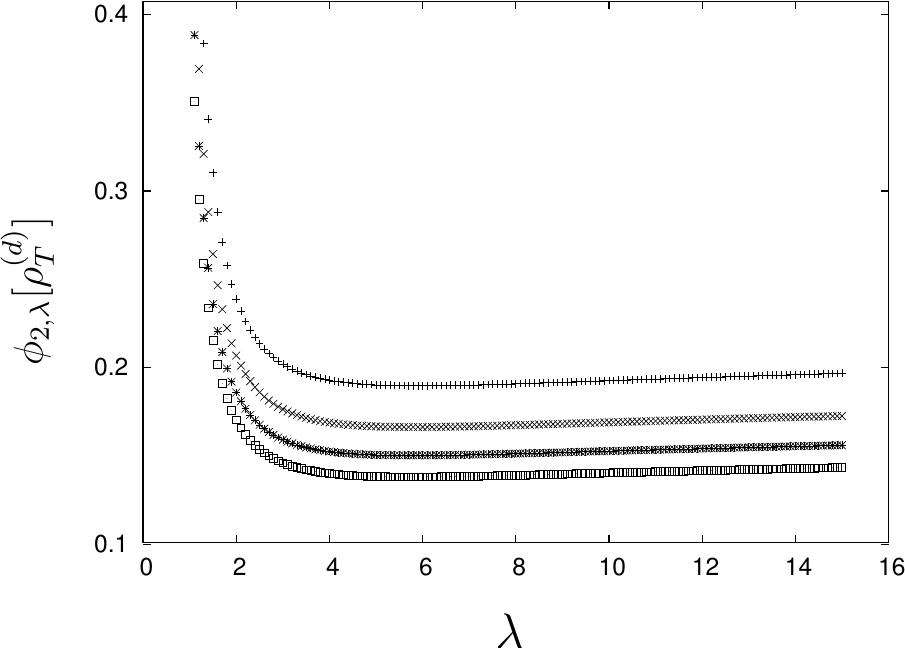}\hfill
       \includegraphics[width=.5\linewidth]{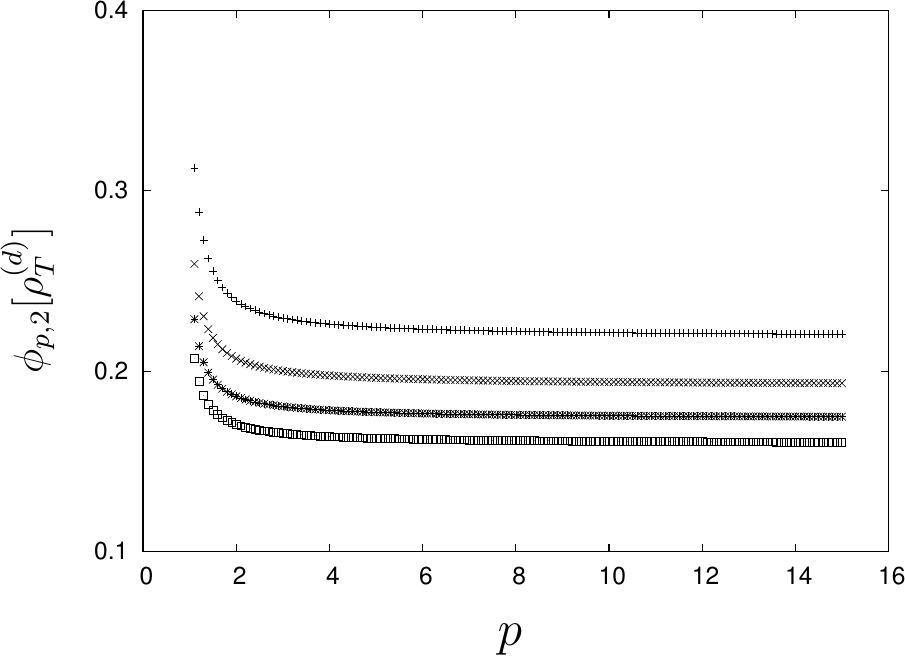}
      \end{minipage}%
      }\par
  \caption{Above: Colour maps of the generalized Fisher information $\phi_{p,\lambda}[\rho_{T}^{(d)}]$ in $\frac{k _BT}{h}$-units against the parameters $(p,\lambda)$ and $(q,\lambda)$ respectively, when $d=6$. Below left: the generalized Fisher information $\phi_{2,\lambda}[\rho_{T}^{(d)}]$ in terms of $\lambda$ for $d=3-6$. Below right: the generalized Fisher information $\phi_{p,2}[\rho_{T}^{(d)}]$ in terms of $p$ for $d=3-6$. In the last two graphs the upper (lower) curve corresponds to the case $d=3\, (d=6).$}
  \label{fig:fisher}
  \end{figure}

\subsection{Biparametric complexity measures}

\subsubsection{Biparametric Crámer-Rao complexity}
Let us now calculate the generalized Crámer-Rao complexity $C_{FR}^{(p,\lambda)}[\rho_{T}^{(d)}]$ of the $d$-dimensional blackbody density at temperature $T$ which, according to (\ref{bcr}), is given by
\begin{equation}
\label{eq:comp_CR1}
C_{CR}^{(p,\lambda)}[\rho_{T}^{(d)}]=\mathcal K_{CR}(p,\lambda)\,\,\phi_{p,\lambda}[\rho_{T}^{(d)}] \,\,  \sigma_{p}[\rho_{T}^{(d)}],\nonumber\\
\end{equation}
where the constant $\mathcal K_{CR}(p,\lambda)$ is given in Eq. (\ref{kcr0}). Taking into account the values of $\phi_{p,\lambda}[\rho_{T}^{(d)}]$ and $\sigma_{p}[\rho_{T}^{(d)}]$ given by Eqs. (\ref{eq:genfish0}) and (\ref{eq:moments}), respectively, we obtain that the complexity measure $C_{FR}^{(p,\lambda)}[\rho_{T}^{(d)}]$ can be expressed as
\begin{eqnarray}
\label{eq:comp_CR2}
C_{CR}^{(p,\lambda)}[\rho_{T}^{(d)}]&=&
\mathcal K_{CR}(p,\lambda)(\Gamma(d+1)\zeta(d+1))^{-\alpha}\nonumber \\
& &\hspace{-1cm}\times \Big[\sum_{i=0}^q(-1)^{q-i}{q\choose i}d ^i(\alpha d -i)!\zeta_{q\lambda+1-i}(1+\alpha d -i ,1+q(\lambda-1))\Big]^{\frac{1}{q\lambda}} \left(\sum_{n=0}^{p}\gamma_n(d,p)\,\,\zeta(d+n+1)\right)^{\frac{1}{p}},
\end{eqnarray}
for even $q$, $q \lambda \in\mathbb{N}$ and where the symbol $\gamma_n(d,p)$ is defined by (\ref{eq:ctemoments}) and $\zeta_m (x,y)$ is the Barnes zeta function mentioned above. Most important is to note that this complexity quantifier depends only on the parameters $(p,\lambda)$ and the dimensionality of the universe $d$. \\
In Fig. \ref{fig:cramerrao}, we plot a colour tridimensional map of $C_{CR}^{(p,\lambda)}[\rho_{T}^{(d)}] \equiv C_{CR}^{(p,\lambda)}(d)$ against the parameters $(p,\lambda)$, and the \textit{conjugated} representation with respect to the parameters $(p,\lambda)$ when $d=6$. We observe that this complexity measure captures a non-trivial structure with an absolute minimum valley. Similar maps can be obtained for other dimensionalities. For completeness let us point out that when $d=3$, the absolute minimum is located at $(p\simeq1.91,\lambda\simeq1.55)$, for which the complexity $C_{CR}^{(1.91,1.55)}\simeq1.29$.
This illustrates to what extent the Crámer-Rao complexity captures such an structure even for distributions so well behaved as the generalized Planck distribution law. This suggests that this complexity quantifier must be a powerful tool for the information-theoretical analysis of much more complex physical laws.\\ 
 
To get a further insight into this complexity map $C_{CR}^{(p,\lambda)}(d)$ we make two cuts at $p=2$ and $\lambda=2$ for various dimensionalities $d=3,4,5,6$ as it is shown in the two below graphs of the figure. In the below-left graph we plot the complexity $C_{CR}^{(3,\lambda)}(d)$ in terms of $\lambda$ for $\lambda>1/3$, finding the existence of a value $\lambda_{min}$ which minimizes this measure; as well, we observe that it tends toward a constant value at large values of $\lambda$. In the below-right graph we plot the complexity $C_{CR}^{(p,2)}(d)$ in terms of $p$, we also find a minimum but, opposite to the previous case, the asymptotic $p$-behavior is clearly divergent.

\begin{figure}[H]
  \subfloat[]{%
  \begin{minipage}{\linewidth}
 \includegraphics[width=.5\linewidth]{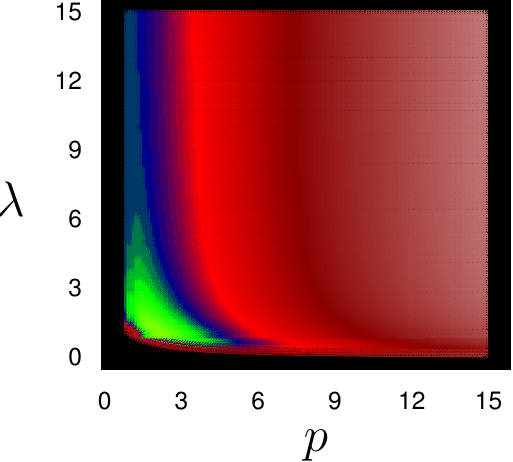}\hfill
   \includegraphics[width=.47\linewidth]{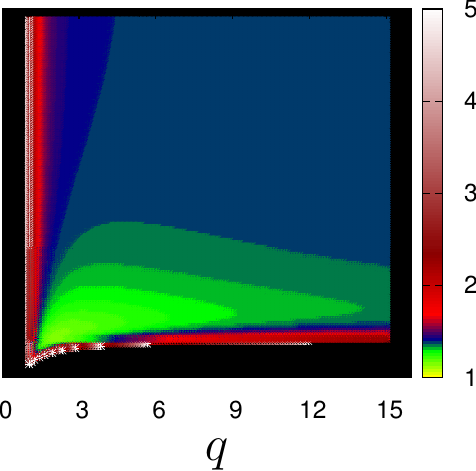}
  \end{minipage}%
  }\par
    \subfloat[]{%
      \begin{minipage}{\linewidth}
     \includegraphics[width=.5\linewidth]{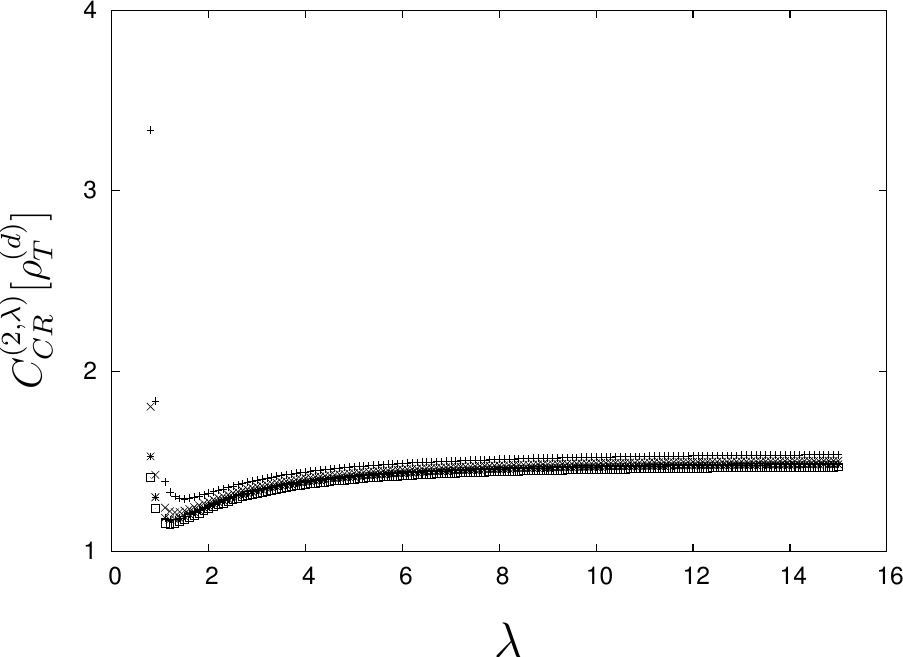}\hfill
       \includegraphics[width=.5\linewidth]{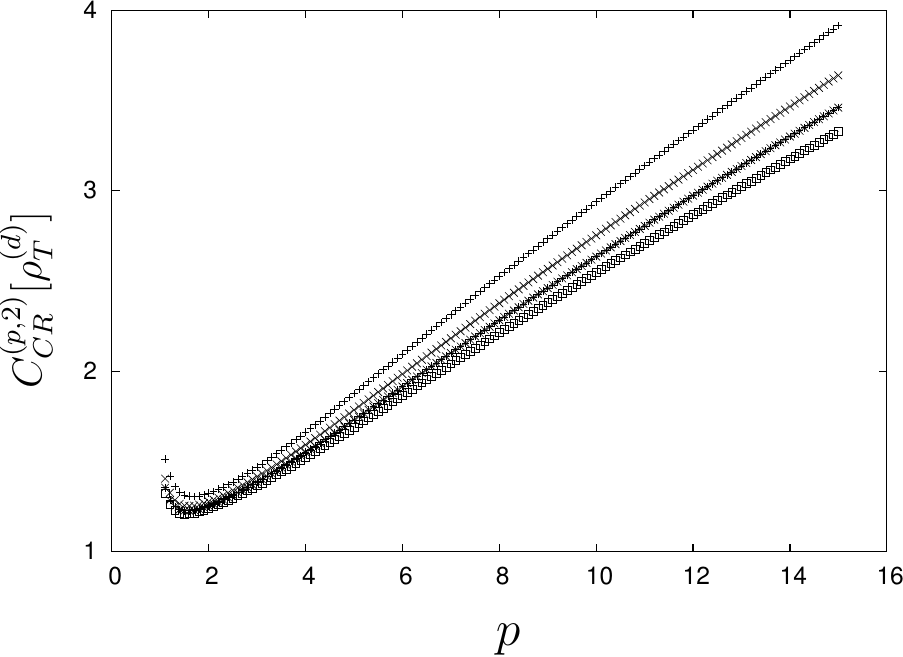}
      \end{minipage}%
      }\par
    \caption{Above: Colour map of the Crámer-Rao complexity $C_{CR}^{(p,\lambda)}(d)$ against the parameters $(p,\lambda)$ and $(q,\lambda)$ when $d=6$. Below left: Dependence of the Crámer-Rao complexity $C_{CR}^{(2,\lambda)}(d)$ on $\lambda$ when $d=6$. Below right: Dependence of the Crámer-Rao complexity $C_{CR}^{(p,2)}(d)$ on $p$ when $d=3-6$. In the last two graphs the upper (lower) curve corresponds to the case $d=3\, (d=6).$}
  \label{fig:cramerrao}
  \end{figure}

\subsubsection{Biparametric Heisenberg-Rényi complexity}
The biparametric Heisenberg-Rényi $C_{HR}^{(p,\lambda)}[\rho_{T}^{(d)}]$ of the $d$-dimensional blackbody density $\rho_{T}^{(d)}$ can be written, according to (\ref{bhr}), as
 \[
 C_{HR}^{(p,\lambda)}[\rho_{T}^{(d)}]= \mathcal K_{HR}(p,\lambda)\frac{\sigma_p[\rho_{T}^{(d)}]}{N_\lambda[\rho_{T}^{(d)}]}
 \]
 with the constant  $\mathcal{K}_{HR}(p,\lambda)$ given by (\ref{khr})
 Moreover, in the general case $\lambda\not=1$ this constant is
 \[
 \mathcal{K}_{HR}(p,\lambda)=(p\lambda+\lambda-1)^{\frac{q\lambda-\lambda+1}{q\lambda-q}} (p\lambda)^{\frac{1}{1-\lambda}} a_{p,\lambda}^{-1},
 \]
 so that the corresponding expression for the complexity measure $C_{HR}^{(p,\lambda)}[\rho_{T}^{(d)}]$ is
 \begin{equation}
 C_{HR}^{(p,\lambda)}[\rho_{T}^{(d)}]=\mathcal{K}_{HR}(p,\lambda)\left(\frac{\Gamma (\lambda d+1)\zeta_\lambda(\lambda d+1,\lambda)}{\Gamma^\lambda (d+1) \zeta^\lambda(d+1)}\right)^{\frac {1}{\lambda-1}} \left(\sum_{n=0}^{p}\gamma_n(d,p)\zeta(d+n+1)\right)^{\frac{1}{p}},
 \end{equation}
with $\lambda \in \mathbb{N}$, $p$ even. Again here, we note that this complexity quantifier depends only on the parameters $(p,\lambda)$ and the dimensionality of the universe $d$. \\

In Fig. \ref{fig:HRs}, a colour tridimensional map of $C_{HR}^{(p,\lambda)}[\rho_{T}^{(d)}] \equiv C_{HR}^{(p,\lambda)}(d)$ is given, which shows the dependence of the Heisenberg-Rényi complexity in terms of the parameter $\lambda$ for different values of the parameter $p$ for the spatial dimensionality $d=6$. We observe that Heisenberg-Rényi complexity measure allows us to capture a non-trivial structure with an absolute minimum. Similar complexity maps can be obtained for other dimensionalities. In particular when $d=3$ this minimum is located at $(p\simeq1.34,\lambda\simeq1.24)$, for which this measure has the value $C_{CR}^{(1.34,1.24)}\simeq1.08$. To better understand this figure at the dimensionalities $d=3,4,5,6$, we make two cuts at $p=3$ and at $\lambda=1$ which give rise to the two below graphs. In the below-left graph we plot the complexity $C_{HR}^{(3,\lambda)}(d)$ in terms of $\lambda$ for $\lambda>1/4$, finding the existence of a $\lambda_{min}$ which minimizes the measure as well as a constant asymptotic trend when $\lambda\to\infty$. In the below-right graph we plot the complexity quantifier $C_{HR}^{(p,1)}(d)$ for $p>0$, finding a minimum value $p_{min}$ which minimizes the complexity, as well as a divergent asymptotic behavior similar to the one previously found for the Crámer-Rao complexity measure.
\begin{figure}[H]
  \subfloat[]{%
  \begin{minipage}{\linewidth}
  \centering
 \includegraphics[width=.52\linewidth]{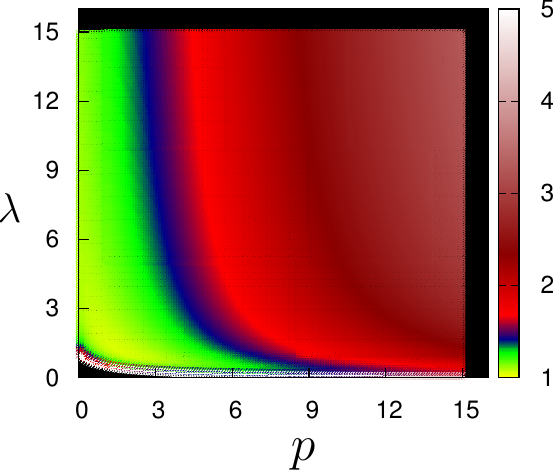}\hfill
  \end{minipage}%

  }\par
   \subfloat[]{%
    \begin{minipage}{\linewidth}
   \includegraphics[width=.5\linewidth]{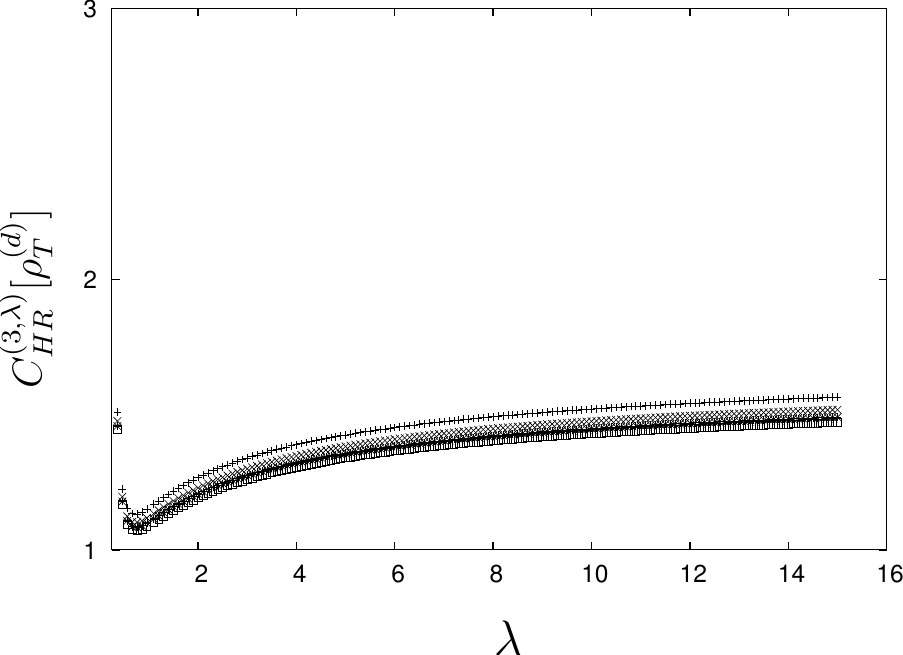}\hfill
     \includegraphics[width=.5\linewidth]{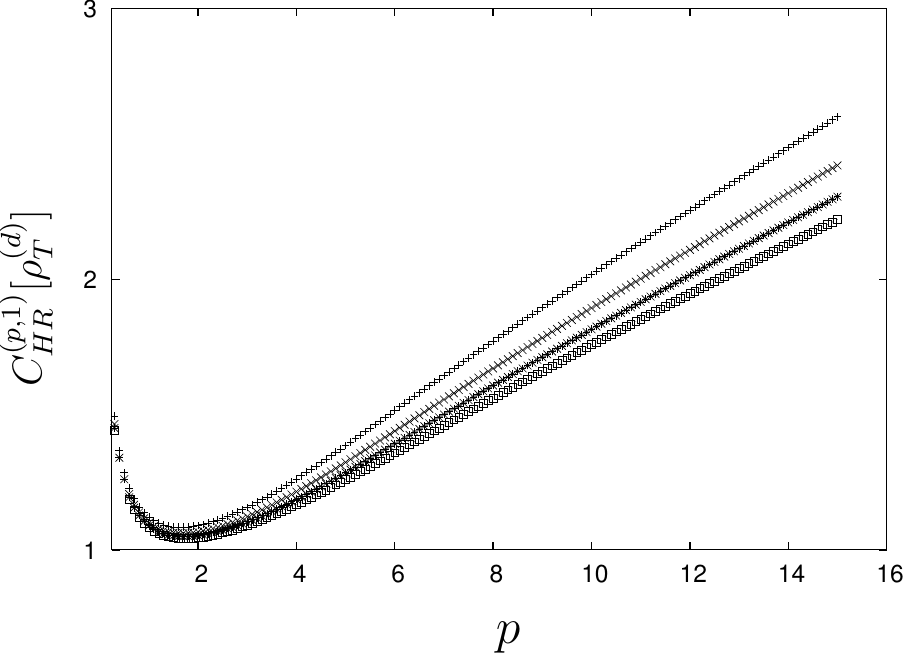}
    \end{minipage}%
    }\par    
    \caption{Above: Colour map of the Heisenberg-Rényi complexity $C_{HR}^{(p,\lambda)}[\rho_{T}^{(d)}] \equiv C_{HR}^{(p,\lambda)}(d)$ against the parameters $(p,\lambda)$ when $d=6$. Below left: Dependence of the Heisenberg-Rényi complexity $C_{HR}^{(3,\lambda)}(d)$ on $\lambda$ when $d=3-6$. Below right: Dependence of the Heisenberg-Rényi complexity $C_{HR}^{(p,1)}(d)$ on $p$ when $d=3-6$. In the last two graphs the upper (lower) curve corresponds to the case $d=3\, (d=6).$}

    \label{fig:HRs}
  \end{figure}
  
  

Finally, for completeness, in Fig. \ref{fig:comparative} the generalized Gaussian distributions which minimize the two novel complexity quantifiers of Crámer-Rao (red color) and Heisenberg-Rényi (blue) types introduced in this work are compared with the corresponding Planck distribution law (black) for the dimensionality $d=3$. We observe certain similarities in the left fall of the Crámer-Rao and Planck cases, and in the right fall of the Heisenberg-Rényi and Planck cases.

\begin{figure}[H]
  \centering
 \includegraphics[scale=.8]{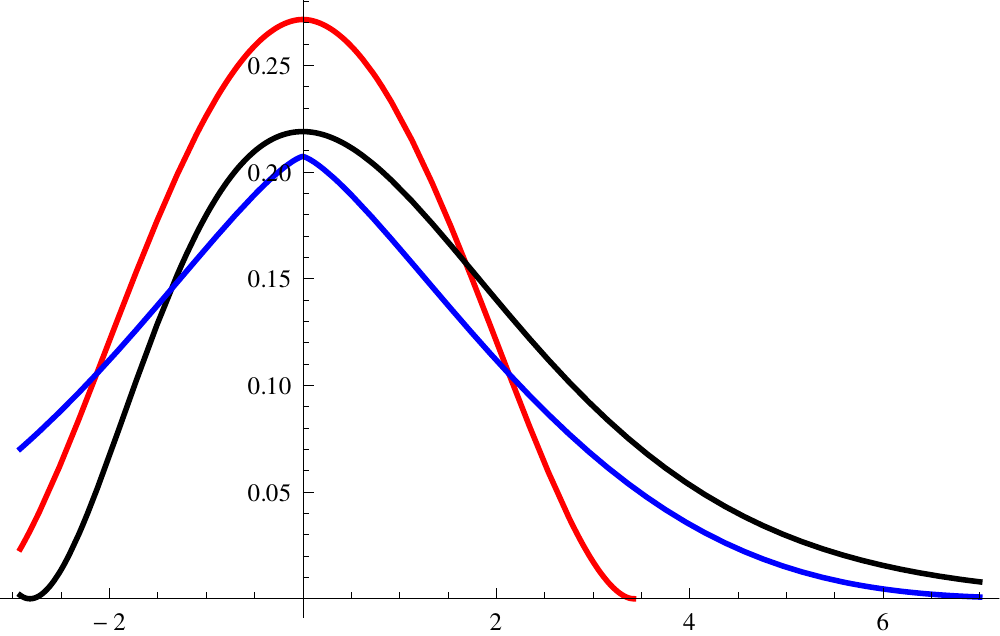}
\caption{Comparison of the generalized Gaussian minimizers of the extended Crámer-Rao (red) and Heisenberg-Rényi (blue) with Plank distribution law (black) when $d=3$.}
\label{fig:comparative}
 \end{figure}
 
\section{Conclusions and open problems}

It is known that we need various measures to take into account the multiple facets of the concept of complexity in a complex many-body system. In this paper we have introduced and discussed two novel biparametric complexity tools of Crámer-Rao and Heisenberg-Rényi types, which extend the three basic  measures of complexity (i.e., Crámer-Rao, Fisher-Shannon and LMC) and some modifications which have been published up until now. Then we have illustrated the usefulness of these two complexity measures by applying and explicitly computing them for a relevant quantum object, the $d$-dimensional blackbody at temperature $T$. We have found that they are universal constants in the sense that they are dimensionless and they do not depend on the temperature nor on any physical constant (such as e.g., Planck constant, speed of light or Boltzmann constant), so that they only depend on the spatial  dimensionality of the universe. The results show the existence of a non trivial underlying mathematical structure, according to which these quantities become minimal for some values of their characteristic parameters.\\

To determine these generalized measures of complexity for the $d$-dimensional blackbody radiation with standard ($d = 3$) and non-standard dimensionalities we needed to calculate various dispersion and entropy-like quantities in terms of dimensionality $d$ and temperature $T$. Indeed, we have determined the typical deviations (that generalize the standard deviation), the Rényi entropy (that generalizes the Shannon entropy and the disequilibrium) and the biparametric Fisher information (which generalizes the standard Fisher information) of the $d$-dimensional Planck density in an analytical way. We have found that these quantities, slightly modified, have a Wien-like temperature behavior similar to the well-known Wien's law followed by the frequency $\nu_{max}$ at which the density is maximum. The values of these characteristic quantities, particularly the ones associated to the biparametric Fisher information, might be of potential interest to grasp the anisotropies of the cosmic microwave background radiation (which yields information about our Universe at around 380 000 years after the Big Bang). Finally, we wonder whether this information-theoretical approach may be used for the (broadly unknown) cosmic neutrino background and the cosmic gravitational background, which would provide hints about our Universe one minute after the Big Bang and during the Big Bang, respectively \cite{Faessler}.


\section*{Acknowledgements} 
This work was partially supported by the Projects
P11-FQM-7276 and FQM-207 of the Junta de Andalucia, and by the MINECO-FEDER (European Regional Development Fund) grants FIS2014-54497P and FIS2014-59311P. I. V. Toranzo acknowledges the support of the Spanish Ministerio de Educación under the program FPU 2014.\\

Statement: All authors contributed equally to the paper.

\newpage
\section{Appendix A}
\appendix
\label{Lemma:app}
Here we explicitly solve the integral functionals needed to determine the Rényi entropies of the $d$-dimensional blackbody in section III.\\
 
\begin{lemma}
Let  $n,m,k \in \mathbb{N}_0, n,k>0$, $n>k\ge m$ y $r,s,p \in \mathbb{R}$, con $r>s, r>p$. Then, the following multiparametric integral has the value
\begin{eqnarray}
\label{eq:intbbd1}
\int_0^\infty \frac{x^ne^{(mr+(k-m)s+p)x}}{(e^{rx}-e^{sx})^{k+1}}\,dx &=& \frac{1}{(r-s)^{n+1}}\frac{n!}{k!} \sum_{i=0}^{k}\sum_{j=0}^{k-i}(-1)^{k+i}{i+j \choose j}S_k^{(i+j)} \left(k-m+\frac{s-p}{r-s}\right)^j \zeta\left(n+1-i,\frac{r-p}{r-s}\right) \nonumber \\
& = &  \frac{1}{(r-s)^{n+1}}n! \sum_{i=0}^{k} q_{k+1,i}\left(k-m+\frac{r-p}{r-s}\right)\, \,\zeta\left(n+1-i,\frac{r-p}{r-s}\right),\nonumber\\
\end{eqnarray}
where the Stirling numbers $S_n^{(l)}$ and the Choi coefficients are related by Eq. (\ref{choicoeff}).
\end{lemma}
\begin{proof}
 Let us begin with the multiparametric functional
\begin{equation}\label{funcional}
_k^tJ_m^n(r,s,p)=\int_0^\infty \frac{x^ne^{(m-k)rx}e^{(m-t)sx}e^{px}}{(e^{rx}-e^{sx})^m}dx
\end{equation}
with  $r>0, r>s, r>p$, y $n,m,k,t\in \mathbb{N}$, $n+1> m\ge k, t$. By deriving this functional with respect to $s$ and $r$, one readily finds some recurrence relations $_k^tJ_m^n(r,s,p)$ for it. 
For convenience, however, we first make the change of variable $y=(r-s)x$, because then one realizes that the functional only depends on $\frac{r-p}{r-s} $  when $m+1=k+t$, so that it is better to write 
\begin{equation}\label{eq:eqf0}
_{k}^{t}J_{k+t-1}^{n}(r,s,p)=\frac{1}{(r-s)^{n+1}}\ _{k}^{t}f_{k+t-1}^{n}\left(\frac{r-p}{r-s}\right),                                                                                                                   
\end{equation}
and then the abovementioned derivations yield the following recurrence relations:
\begin{eqnarray}
 \label{eq:eqf1}
_{k+1}^1f_{k+1}^{n+1}(x)&=&\frac{1}{k}\Big[\left(n+1+x\frac{d}{dx}\right)\ _{k}^1f_k^n(x)-(k-1)_{k}^1f_k ^{n+1}(x) \Big]\\
\label{eq:eqf2}
_k^{t+1}f_{k+t}^{n+1}(x)&=&\frac{1}{k+t-1}\Big[\left(n+1+(x-1)\frac{d}{dx}\right)\ _k^{t}f_{k+t-1}^n(x)+(t-1)_k^tf_{k+t-1}^{n+1}(x)\Big].
\end{eqnarray}
On the other hand we can obtain that
\begin{equation}\label{digamma}
_1^1f_1^n(r,s,p)=(-1)^{n+1}\psi^{(n)}\left(\frac{r-p}{r-s}\right),
\end{equation}
by noticing that $_1^1J_1^n(r,s,p)$ is the $n$-th derivative of the integral (see Eq. 3.311-11 of Ref. \cite{ryzhik})
\begin{eqnarray}
\nonumber
\int_0^\infty \frac {e^{px}-e^{qx}}{e^{rx}-e^{sx}} dx =\frac{1}{r-s}\left[\psi\left(\frac{r-q}{r-s}\right)-\psi\left(\frac{r-p}{r-s}\right)\right],
\end{eqnarray}
with respect to $p$. The symbol $\psi^{(n)}(x)$ denotes the $n$th derivative of the digamma function \cite{olver}.\\

The recurrence relation (\ref{eq:eqf1}) with the initial condition (\ref{digamma}) gives rise by induction to
\begin{eqnarray} \label{eq:eqf3}
_k^1f_k^n(x)&=&\frac{(-1)^{k+n}}{(k-1)!}\sum_{j=0}^{k-1}\sum_{i=0}^{k-j-1}{n \choose j}\frac{(i+j)!}{i!}S_{k-1}^{(i+j)}(x+k-2)^{i} \psi^{(n-j)}(x).
\end{eqnarray}
Then, the recurrence relation (\ref{eq:eqf3}) in $t$ with the initial ($t=0$) condition allows us to obtain also by induction the expression
\begin{eqnarray}\label{eq:eqf4}
_k^{t+1}f_{k+t}^n(x) &=&\frac{(-1)^{k+t+n}}{(k+t-1)!} \sum_{j=0}^{k+t-1}\sum_{i=0}^{k+t-j-1}{n\choose j}S_{k+t-1}^{(i+j)}\frac{(i+j)!}{i!}(x+k-2)^{i}\ \psi^{(n-j)}(x). 
\end{eqnarray}
Now the replacement of (\ref{eq:eqf4}) into (\ref{eq:eqf0}), taking into account that $\psi^{(n)}(x)=(-1)^{n+1}n!\zeta(n+1,x)$ and redefining the involved parameters in a convenient manner, we finally obtain the wanted expression (\ref{eq:intbbd1}): 
\begin{eqnarray*}
\int_0^\infty \frac{x^ne^{mrx}e^{(k-m)sx}e^{px}}{(e^{rx}-e^{sx})^{k+1}}\,dx &=&  \frac{1}{(r-s)^{n+1}}\frac{n!}{k!}   \sum_{i=0}^{k}\sum_{j=0}^{k-i}(-1)^{k+i}{i+j \choose j}S_k^{(i+j)}\left(k-m+\frac{s-p}{r-s}\right)^j   \zeta\left(n+1-i,\frac{r-p}{r-s}\right) 
\end{eqnarray*}
where $n,m,k \in \mathbb{N}$, $n>k\ge m$ y $r,s,p \in \mathbb{R}$, con $r>s, r>p$. And from this expression and Eq. (\ref{choicoeff}) follows the second expression of the Lemma.
\end{proof}

\begin{corollary}
Let $k\in \mathbb{N}$, $a\in\mathbb{R}$, $n\in \mathbb{N}$. Then, the following finite sum of standard Hurwitz functions $\zeta(s,a)$  
 \begin{equation}
 Z_k(n,a,t)=\sum_{i=0}^{k-1}q_{k,i}(a)\zeta(n-i,t),
\end{equation}
verifies
\begin{align}
 \zeta_{k}(n,a)&=Z_k(n,a,\{a\})=Z_k(n,a,1+\{a\})=\ldots=Z_k(n,a,a), 
\end{align}
$\forall a\in\mathbb{R}/\mathbb{N}$ (with $\{a\}\equiv a-[a]$ being the non-integer part of $a$),and
\begin{align}
 \zeta_{k}(n,a)&=Z_k(n,a,1)=Z_k(n,a,2) =\ldots=Z_k(n,a,a), \quad \forall a\in\mathbb{N}.
\end{align}
\end{corollary}
\begin{proof}
Using the previous Lemma with $s=0$, $r=1$ and $p<1$ for $k,m,n \in \mathbb{N}_0$ and $n> k\ge m$ one has that
\begin{eqnarray}\label{repint}
\frac 1{n!}\int_0^\infty \frac{x^{n}e^{(m+p)x}}{(e^x-1)^{k+1}}dx &=& \sum_{i=0}^{k}q_{k+1,i}(k-m-p+1)\zeta(n+1-i,1-p). 
\end{eqnarray}
On the other hand, taking into account the integral representation \cite{barnes} with $m+p=k+1-a$ we can write 
\begin{eqnarray}
 \zeta_{k+1}(n+1,a)&=&\frac 1{\Gamma(n+1)}\int_0^\infty \frac{x^{n}e^{(m+p)x}}{(e^x-1)^{k+1}}\,dx = \sum_{i=0}^{k}q_{k+1,i}(a)\zeta(n+1-i,1-p)
\end{eqnarray}
(where the last identity holds provided that $1>p=k+1-m-a\ge 1-a$) and using the notation $a-[a]\equiv\{a\}$, where $[a]$ denotes the integer part of $a\in \mathbb{R}$, it is straightforward to see that $p+\{a\}=k+1-m-[a]\equiv n' \in \mathbb{N}$. The latter implies that, due to the conditions $0\le m \le k$ and $p<1$, the values of $p$ are limited to $p= 1-a,\, 2-a,...,[a]+1-a<1, \,\, \forall a\in\mathbb{R}/\mathbb{N}$; for $a\in\mathbb{N}$ the inequality is fulfilled for $p=1-a, \cdots, [a]-a=0<1$. 
Thus, we have proved that
$$ \zeta_{k+1}(n+1,a)=\sum_{i=0}^{k}q_{k+1,i}(a)\zeta(n+1-i,1-p)$$
where $p$ can take the values $p = 1-a,\, 2-a,...,\,1-\{a\}$, save when $a\in\mathbb N$ in which case  $p = 1-a,\, 2-a,...,\,0$ so that then one has
\begin{eqnarray}
\label{eq:Z}
 \zeta_{k}(n,a)&=&Z_k(n,a,\{a\})=Z_k(n,a,1+\{a\}) =Z_k(n,a,2+\{a\}) = \ldots =Z_k(n,a,a)
\end{eqnarray} 
$\forall a\in\mathbb{R}/\mathbb{N}$, and
\begin{eqnarray}
\label{eq:ZN}
 \zeta_{k}(n,a)&=&Z_k(n,a,1)=Z_k(n,a,2)  = \ldots =Z_k(n,a,a)\quad \forall a\in\mathbb{N}.
\end{eqnarray} 
\end{proof}
\begin{corollary}
For  $k \in \mathbb{N},n \in \mathbb{N},m \in \mathbb{N}_0$, and $n+1> k> m  \ge 0$ one has that
\[
\frac 1 {n!}\int_0^\infty \frac{x^ne^{mx}}{(e^{x}-1)^{k}}dx=
 \sum_{i=0}^{k-1}q_{k,i}(k-m)\, \zeta(n+1-i).
\]
\end{corollary}
This result directly follows from the previous Lemma with $r=1$ and $s=p=0$.

\end{document}